\documentclass[10pt,onecolumn]{IEEEtran}
\makeatletter
\def\ps@headings{%
\def\@oddhead{\mbox{}\scriptsize\rightmark \hfil \thepage}%
\def\@evenhead{\scriptsize\thepage \hfil \leftmark\mbox{}}%
\def\@oddfoot{}%
\def\@evenfoot{}}
\makeatother
\usepackage{amsmath}
\usepackage{amsthm}
\usepackage{amssymb}
\usepackage{verbatim}
\usepackage{algorithm}
\usepackage{algorithmic}
\usepackage{graphicx}
\usepackage{paralist}
\usepackage{color}
\usepackage{hyperref}

\usepackage{xcolor,colortbl}
\usepackage{verbatim}
\usepackage{algorithm}
\usepackage{algorithmic}
\usepackage{graphicx}
\usepackage{epstopdf}  
\usepackage{paralist}
\usepackage{color}
\usepackage{subfigure}

\newtheorem{lemma}{Lemma}
\newtheorem{theorem}{Theorem}
\newtheorem{definition}{Definition}

\floatname{algorithm}{Algorithm}
\newtheorem{remark}{Remark}

\title{Stability Analysis of Frame Slotted Aloha Protocol}
\author{Jihong Yu, Lin Chen \\
LRI-CNRS UMR 8623, Univ. Paris-Sud, 91405 Orsay, France. \{jihong.yu, chen\}@lri.fr}
\begin{document}

\maketitle
\begin{abstract}
Frame Slotted Aloha (FSA) protocol has been widely applied in Radio Frequency Identification (RFID) systems as the \textit{de facto} standard in tag identification. However, very limited work has been done on the stability of FSA despite its fundamental importance both on the theoretical characterisation of FSA performance and its effective operation in practical systems. In order to bridge this gap, we devote this paper to investigating the stability properties of FSA by focusing on two physical layer models of practical importance, the models with single packet reception and multipacket reception capabilities. Technically, we model the FSA system backlog as a Markov chain with its states being backlog size at the beginning of each frame. The objective is to analyze the ergodicity of the Markov chain and demonstrate its properties in different regions, particularly the instability region. By employing drift analysis, we obtain the closed-form conditions for the stability of FSA and show that the stability region is maximised when the frame length equals the backlog size in the single packet reception model and when the ratio of the backlog size to frame length equals in order of magnitude the maximum multipacket reception capacity in the multipacket reception model. Furthermore, to characterise system behavior in the instability region, we mathematically demonstrate the existence of transience of the backlog Markov chain.
\end{abstract}


\section{Introduction}

\subsection{Context and Motivation}

Since the introduction of Aloha protocol in 1970~\cite{Abramson1970}, a variety of such protocols have been proposed to improve its performance, such as Slotted Aloha (SA)~\cite{Lawrence1975} and Frame Slotted Aloha (FSA)~\cite{Okada1977}. SA is a well known random access scheme where the time of the channel is divided into identical slots of duration equal to the packet transmission time and the users contend to access the server with a predefined slot-access probability. As a variant of SA, FSA divides time-slots into \textit{frames} and a user is allowed to transmit only a single packet per frame in a randomly chosen time-slot.

Due to the effectiveness to tackle collisions in wireless networks, Aloha-based protocols have been applied extensively to various networked systems ranging from the traditional satellite networks~\cite{Okada1977}, wireless LANs \cite{Vasudevan2009} to the emerging Machine-to-Machine (M2M) communications~\cite{Wu2013}. Specifically, in radio frequency identification (RFID) systems, FSA plays a fundamental role in the identifications of tags~\cite{Zhu2010,Zhu2011} and is standardized in EPCGlobal Class-1 Generation-2 (C1G2) RFID standard~\cite{EPCGlobal2005}. In RFID systems, all tags transmit in the first frame in the selected slot respectively, but only tags experiencing no collisions are identified while other nodes, referred to as backlogged nodes (or simply backlogs), retransmit in the subsequent frames until all of them are successfully identified.

It is evident that stability is of primary importance for systems operating on top of Aloha-like protocols. A large body of studies have been devoted to stability analysis in a slotted collision channel \cite{Szpankowski1983,Rosenkrantz1983,Pountourakis1992} where a transmission is successful if and only if just a single user transmits in the selected slot, referred to as single packet reception (SPR).

More recently, application of FSA in RFID systems has received considerable research attention. However, very limited work has been done on the stability of FSA despite its fundamental
importance both on the theoretical characterisation of FSA performance and its effective operation in practical systems. Moreover, the emerging multipacket reception (MPR) technologies in wireless networks, such as Code Division Multiple Access (CDMA) and Multiple-Input and Multiple-Output (MIMO), make it possible to receive multiple packets in a time-slot simultaneously, which remarkably boosts system performance. Motivated by the above observation, we argue that a systematic study on the stability properties of FSA incorporating the MPR capability is called for in order to lay the theoretical foundations for the design and optimization of FSA-based communication systems.

\subsection{Summary of Contributions}

In this paper, we investigate the stability properties of FSA with single packet reception and multipacket reception capabilities. The main contributions of this paper are articulated as follows:
\begin{itemize}
\item  Firstly, we model the packet transmission process in a frame as the bins and balls problem~\cite{johnson1977} and
       derive the number of successful received packets under both SPR and MPR models.
\item  Secondly, we formulate a homogeneous Markov chain to characterize the number of the backlogged packets and derive the one-step transition probability.
\item  Thirdly, by employing drift analysis, we obtain the closed-form conditions for the stability of FSA and derive conditions maximising the stability regions for both SPR and MPR models.
\item  Finally, to characterise system behavior in the instability region, we mathematically demonstrate the existence of transience of the backlog Markov chain.
\end{itemize}

Our work demonstrates that the stability region is maximised when the frame length equals the backlog size in the MPR model and when the ratio of the backlog size to frame length equals in order of magnitude the maximum multipacket reception capacity in the MPR model. In addition, it is also shown that FSA-MPR outperforms FSA-SPR remarkably in term of the stability region size.

\subsection{Paper Organisation}
The remainder of the paper is organised as follows.
Section \ref{sec:related_work} gives a brief overview of related work and compares our results with existing results.
In Section \ref{sec:SM}, we present the system model, including random access model, traffic model and packet success probability.
In Section \ref{sec:MR}, we summary the main result of this paper. Then the detailed proofs on the stability properties of FSA-SPR and FSA-MPR are given in
Section \ref{sec:Sta-SPR} and
Section \ref{sec:Sta-MPR}, respectively.
Finally, we conclude our paper in Section~\ref{sec:conclusion}.

\section{Related Work}
\label{sec:related_work}

Aloha-based protocols are basic schemes for random medium access and are applied extensively in many communication systems. As a central property, the stability of Aloha protocols has received a lot of research attention, which we briefly review in this section.

\textbf{Stability of slotted Aloha} Tsybakov and Mikhailov \cite{Tsybakov1979} initiated the stability analysis of finite-user slotted Aloha in 1979. They found sufficient conditions for stability of the queues in the system using the principle of stochastic dominance and derived the stability region for two users explicitly. For the case of more than two users, the inner bounds to the stability region were shown in \cite{Rao1988}. Szpankowski \cite{Szpankowski1994} found necessary and sufficient conditions for the stability of queues for a fixed transmission probability vector assuming the arrival rates follow Bernoulli process. However, the derived conditions are not closed-form, meaning that verifying them is difficult. Borst~\textit{et al.} extended the results of~\cite{Szpankowski1994} to more general systems assuming the monotonicity of the service rate~\cite{Borst2008}. We would like to point out that all the above stability analysis results were derived for the SPR model.

\textbf{Stability of slotted Aloha with MPR} The first attempt at analyzing stability properties of SA with MPR was made by Ghez~\textit{et al.} in \cite{Ghez1988,Ghez1989} in an infinite-user single-buffer model. They drew a conclusion that the system could be stabilized under the symmetrical MPR model with a non-zero probability that all packets were transmitted successfully. Sant and Sharma \cite{Sant2000} studied a special case of the symmetrical MPR model for finite-user with an infinite buffer case. They derived sufficient conditions on arrival rate for stability of the system under the stationary ergodic arrival process. Although the work aforementioned analyzed the stability of system without MPR or/and with MPR, they are mostly, if not all, focused on SA protocol, while our focus is FSA with both SPR and MPR.

\textbf{Performance analysis of FSA} There exist several studies on the performance of FSA. Schoute \cite{Schoute1983} investigated dynamic FSA and obtained the expected number of time-slots needed until the backlog becomes zero. The optimal frame setting for dynamic FSA was proved mathematically by Luca and Flaminio \cite{Barletta2012}. Wieselthier and Anthony \cite{Wieselthier1989} introduced an combinational technique to analyse performance of FSA-MPR for the case of finite users. However, these works did not address the stability of FSA, which is of fundamental importance.

In summary, only very limited work has been done on the stability of FSA despite its fundamental
importance both on the theoretical characterisation of FSA performance and its effective operation in practical systems. In order to bridge this gap, we devote this paper to investigating
the stability properties of FSA under both SPR and MPR models.

\section{System Model}
\label{sec:SM}

In this section, we introduce our system model which will be used throughout the rest of this paper.

\subsection{Random access model in FSA}

We consider a system of infinite identical users operating on one frequency channel. In one slot, a node can complete a packet transmission. FSA organises time-slots with each frame containing $L$ consecutive time-slots. Each user is allowed to randomly and independently choose a time-slot to send his packet at most once per frame.

A packet suffers a collision if more than one packet is transmitted simultaneously in the same time-slot with SPR and if more than $M$ packets are transmitted simultaneously with MPR, where $M$ quantifies the MPR capacity.
In our analysis, we do not distinguish packets, i.e., newly generated packets are simply regarded and treated as backlogs.

\subsection{Traffic model}

Let $N_i$ denote the total number of new packets arriving during frame $ i $ and $A_{il}$ denote the number of new packets arriving during time-slot $l$ in frame $i$ where $l={1,2,\cdots,L}$. Assume that $(A_{il})$ are independent and identically distributed random variables with probability distribution:
\begin{equation}
P\{A_{il} = k\} = \Lambda_k    (k \geq 0)
\end{equation}
such that the expected number of arrivals per time-slot $\Lambda = \sum_{1}^{\infty}{k \Lambda_k}$ is finite.

Then as  $N_i = \sum_{l=1}^{L} {A_{il}}$, the distribution of $N_i$, denoted as $\{\lambda_n\}_{n \ge 0}$, is
\begin{align}
\lambda_n
= P\{N_i = n\} = P\left\{\sum_{l=1}^{L} {A_{il} = n}\right\}
= \sum_{k_1} \sum_{k_2} \cdots \sum_{k_L} {\prod_{l}\Lambda_{k_l}},
\end{align}
where $\sum_{l=1}^L {k_l} = n$. The expected number of arrivals during a frame, denoted as $\lambda$, is $\lambda = L \Lambda$.
%

\subsection{Packet success probability}

The process of randomly and independently choosing a time-slot in a frame to transmit packets can be cast into a class of problems that are known as occupance problems, or bins and balls problem~\cite{johnson1977}. Specifically, consider the setting where a number of balls are randomly and independently placed into a number of bins, the classic occupance problem studies the maximum load of an individual bin.

In our context, time-slots and packets to be transmitted in a frame can be cast into bins and balls, respectively. Given $h$ packets being sent in frame $i$ and the frame length $L$, the number $x$ of packets sent in one time-slot, referred as to occupancy number, is binomially distributed with parameters $h$ and $\frac{1}{L}$:
\begin{equation}
\label{eq:bino}
B_{h,\frac{1}{L}} (x)
= \binom {h}{x} (\frac{1}{L})^x (1 - \frac{1}{L})^{h-x}.
\end{equation}
Applying the distribution of equation~\eqref{eq:bino} to all $L$ slots in the frame, we can get the expected value $b(x)$ of the number of time-slots with occupance number $x$ in a frame as follows:
\begin{align}
b(x)= L \cdot B_{h,\frac{1}{L}} (x) = L \binom {h}{x} (\frac{1}{L})^x (1 - \frac{1}{L})^{h-x}.
\end{align}

We further derive the probability that a packet is transmitted successfully under both SPR and MPR.

\noindent\textbf{Packet success probability of FSA with SPR}

with SPR, the number of successfully received packets equals that of time-slots with occupance number $x=1$. Extending the result of \cite{Vogt2002}, we can derive the probability $\xi_{hk}$ that there exists exactly $k$ successful time-slots in the frame as follows:
\begin{eqnarray}
\xi_{hk} =
\begin{cases}
\frac {\binom{L}{k} \binom{h}{k} k! {G(L-k,h-k)}} {L^h}, & 0 < k < \min(h,L) \\
\frac{\binom{L}{h} h!}{L^h}, & k=h \leq L \\
0, & k > \min(h,L) \\
0, & k = L < h
\end{cases}
\end{eqnarray}
where
\begin{align*}
G(V,u) =
         V^u & +  \sum_{t=1}^{u}{(-1)^{t}} \prod_{j=0}^{t-1} [(u-j)(V-j)] (V-t)^{u-t} \frac{1}{t!},
\end{align*}
with $V\triangleq L-k$ and $u\triangleq h-k$.

Consequently, the expected number of successfully received packets in the frame, denoted as $r_h$, is
\begin{equation*}
r_h
= \sum_{k=1}^{h} {k \xi_{hk}}
= b(1).
\end{equation*}

\noindent\textbf{Packet success probability of FSA with MPR}

With MPR, the number of successfully received packets in a frame is the sum of packets in time-slots with occupance number $1 \leq x \leq M$. Here, we provide the formulation of the probability $\xi_{hk}$ that exact $k$ packets are received successfully among $h$ transmitted packets in the frame. Let occupancy numbers $x_l$ stands for the number of packets in the $l$th time-slot, $l=1,2,\cdots,L$. Every $L$-tuple of integers satisfying
\begin{equation}
\label{eq:configuration}
\sum_{l=1}^{L} x_l = h
\end{equation}
describes a possible configuration of occupancy numbers. Let $Z$ be the set of all possible configurations of occupancy numbers. Denote by $Z_a \in Z$ the configuration with the occupancy numbers ${x_1}^{Z_a},\cdots,{x_L}^{Z_a}$ satisfying \eqref{eq:configuration} and the following conditions:
\begin{eqnarray*}
\begin{cases}
\sum_{l=1}^{L} {{x_l}^{Z_a} Y_l} = k, \\
P(Z_a) = \frac{h!}{ {{x_1}^{Z_a}}! {{x_2}^{Z_a}}! \cdots {{x_L}^{Z_a}}! } {L^{-h}},
\end{cases}
\end{eqnarray*}
where $Y_l$ denotes the indicator function satisfying
\begin{eqnarray*}
Y_l=
\begin{cases}
1& 1 \leq {x_l}^{Z_a} \leq M,\\
0& \text{otherwise}.
\end{cases}
\end{eqnarray*}

Therefore, the probability $\xi_{hk}$ can be written as follows:
\begin{equation}
\label{eq:Suc_pro_MPR}
\xi_{hk} =
          \sum_{Z_a \in Z} P(Z_a)
\end{equation}

Consequently, we can derive the expected number of successfully received packets in the frame as
\begin{align}
\label{eq:Suc_MPR}
r_h= \sum_{k=1}^{h} {k \xi_{hk}} = L \sum_{x=1}^{M} x \binom {h}{x} (\frac{1}{L})^x (1 - \frac{1}{L})^{h-x}.
\end{align}

\section{Main results}
\label{sec:MR}

To streamline the presentation, we summarize the main results in this section and give the detailed proof and analysis in the sections that follow.

Aiming at studying the stability of FSA, we decompose our global objective into the following three questions, all of which are of fundamental
importance both on the theoretical characterisation of FSA performance and its effective operation in practical systems:
\begin{itemize}
\item \textbf{Q1}: Under what condition(s) is FSA stable?
\item \textbf{Q2}: When is the stability region maximised?
\item \textbf{Q3}: How does FSA behave in the instability region?
\end{itemize}

Before answering the questions, we first introduce the formal definition of stability employed by Ghez \textit{et al.} in~\cite{Ghez1988}.

Define by a variable $X_i$ the number of backlogged packets in the system at the start of frame $i$. The discrete-time process $(X_i)_{i \geq 0}$ can be seen as a homogeneous Markov chain.

\begin{definition}
\label{def:sta}
An FSA system is stable if $(X_i)_{i \geq 0}$ is ergodic and unstable otherwise.
\end{definition}

By Definition~\ref{def:sta}, we can transform the study of stability of FSA into investigating the ergodicity of the backlog Markov chain. The rationality of this transformation is two-fold. One interpretation is the property of ergodicity that there exists an unique stationary distribution of a Markov chain if it is ergodic. The other can be interpreted from the nature of ergodicity that each state of the Markov chain can recur in finite time with probability $1$.

We then establish the following results characterizing the stability region and demonstrating the behavior of the backlog Markov chain in nonergodicity regions under both SPR and MPR. For notation convenience, we use FSA-SPR and FSA-MPR to denote the FSA system operating with SPR and MPR, respectively.

\subsection{Results for FSA-SPR}

\begin{theorem}
\label{theorem:stable}
Under FSA-SPR, consider an irreducible and aperiodic backlog Markov chain $(X_i)_{i \geq 0}$ with nonnegative integers. Denote by $h$ the number of backlogged packets in frame $i$ and $\alpha \triangleq \frac{h}{L}$, for $h \to \infty$, we have
\footnote{For two variables X, Y, we use the following asymptotic notations:
\begin{itemize}
\item $X = o(Y)$ if $\frac{X}{Y} = 0$, as $Y \to  \infty$;
\item $X = O(Y)$ if $\frac{X}{Y} = \infty$, as $Y \to  \infty$;
\item $X = \Theta(Y)$ if $\theta_1 \leq \frac{X}{Y} \leq \theta_2$, as $Y \to  \infty$, where $\theta_2 \geq \theta_1 >0$.
\end{itemize}}
\begin{enumerate}\setlength{\parskip} {4pt}
\item The system is always stable for all arrival distributions if $\Lambda < \alpha e^{-\alpha}$ and $L = \Theta(h)$. Specially, the stability region is maximised when $\alpha = 1$.
\item The system is unstable for all arrival distributions under each of the following three conditions: (a) $L = o(h)$; (b) $L = O(h)$; (c) $L = \Theta(h)$ and $\Lambda > \alpha e^{-\alpha}$.
\end{enumerate}
\end{theorem}

\begin{remark}
Theorem~\ref{theorem:stable} answers the first two questions and can be interpreted as follows:
\begin{itemize}
\item When $L=o(h)$, i.e., the backlog size $h$ is far larger than the frame length $L$, a packet experiences collision with high probability (w.h.p.), thus increasing the backlog size and destabilising the system;
\item When $L = O(h)$, i.e., the backlog size is far smaller than the frame length, a packet is transmitted successfully w.h.p.. However, the number of successful transmitted packets is still significantly less than that of new arrivals in the frame. The system is thus unstable.
\item When $L = \Theta(h)$, i.e., the backlog size has the same order of magnitude as the frame length, the system is stable when the backlog can be reduced gradually, i.e., when the expected arrival rate is less than the successful rate.
\end{itemize}
\end{remark}

It is well known that an irreducible aperiodic Markov chain falls into one of three mutually exclusive classes: positive recurrent, null recurrent, and transient. So, our next step after deriving the stability conditions is to show whether the backlog Markov chain in the instability region is transient or recurrent, which answers the third question. To facilitate the demonstration, we focus on the Poisson arrival processes. However, our analysis can be extended to other arrival processes.

\begin{theorem}
\label{th:NON_oh}
With the same notations as in Theorem \ref{theorem:stable} under Poisson arrivals, $(X_i)_{i \geq 0}$ is always transient in the instability region, i.e., under each of the following three conditions: (1) $L = o(h^{1-\epsilon})$, $\forall \ 0< \epsilon <1$; (2) $L = \Theta(h)$ and $\Lambda > \alpha e^{- \alpha}$; (3) $L = O(h)$.
\end{theorem}

\begin{remark}
If a state of a Markov chain is transient, then the probability of returning to itself for the first time in an finite time is less than $1$. Hence, Theorem~\ref{th:NON_oh} implies that once out of the stability region, the system is not guaranteed to return to stable state in a finite time.
\end{remark}

\subsection{Results for FSA-MPR}


\begin{theorem}
\label{theorem:stable-MPR}
Consider an FSA-MPR system where a receiver can decode at most $M$ simultaneously transmitted packets. Using the same notations as in Theorem \ref{theorem:stable}, we have
\begin{enumerate}
\item The system is always stable for all arrival distributions if $\Lambda < \sum_{x=1}^{M} e^{-\alpha} \frac{{\alpha ^x}}{(x-1)!}$ and $L = \Theta(h)$. Specially, let $\alpha^*$ denote the value of $\alpha$ that maximises the stability region, it holds that $\alpha^* =\Theta(M)$.
\item The system is unstable for all arrival distributions under each of the following conditions: (1) $L = o(h)$; (2) $L = O(h)$; (3) $L = \Theta(h)$ and $\Lambda > \sum_{x=1}^{M} e^{-\alpha} \frac{{\alpha ^x}}{(x-1)!}$.
\end{enumerate}
\end{theorem}

\begin{remark}
Comparing the results of Theorem~\ref{theorem:stable-MPR} to Theorem~\ref{theorem:stable}, we can quantify the performance gap between FSA-SPR and FSA-MPR in terms of stability. For example, when $\alpha =1$, the stability region is maximised in FSA-SPR with $\Lambda^{SPR} < e^{-1}$, while the stability region in FSA-MPR is $\Lambda^{MPR} < e^{-1} \sum_{x=1}^{M} \frac{1}{(x-1)!}$. Note that for $M>2$, it holds that
\begin{equation*}
1+1+\frac{1}{2}
< \sum_{x=1}^{M} \frac{1}{(x-1)!}<1+1+\sum_{x=1}^M \frac{1}{x(x+1)}
< 2+\left(\sum_{x=1}^M \frac{1}{x}-\frac{1}{x+1}\right)=3-\frac{1}{M}.
\end{equation*}
The maximum stability region of FSA-MPR is thus between $2.5$ and $3$ times that of FSA-SPR.
\end{remark}

\begin{theorem}
\label{th:NON_oh_OH-MPR}
With the same notations as in Theorem~\ref{theorem:stable-MPR} under Poisson arrivals,
$(X_i)_{i \geq 0}$ is transient under each of the following three conditions: (1) $L = o(h^{1 - \epsilon})$, $\forall \ 0< \epsilon <1$;
      (2) $L = \Theta(h)$ and $\Lambda > \alpha$; (3) $L = O(h)$.
\end{theorem}

\begin{remark}
Theorem~\ref{th:NON_oh_OH-MPR} demonstrates that despite the stability gain of FSA-MPR over FSA-MPR, the properties in the unstable region is almost the same.
\end{remark}


\section{Stability Analysis of Frame Slotted Aloha with Single packet Reception}
\label{sec:Sta-SPR}
In this section, we will analyse the stability of FSA-SPA and prove Theorem~\ref{theorem:stable} and~\ref{th:NON_oh}.
\subsection{Characterising backlog Markov chain}

As mentioned in Sec.~\ref{sec:MR}, we characterize the number $(X_i)_{i \geq 0}$ of the backlogged packets in the system at the beginning of frame $i$ as a homogeneous Markov chain. We now calculate the one-step transition probability as a function of $\xi_{h,k}$ and $\{\lambda_n \}_{n \geq 0}$. Denoting by $h$ the number of the backlogged packets in frame $i$, we obtain the one-step transition probability $\mathbf{P}=\{P_{hk}\}$ as follows:
\begin{itemize}
\item For $h=0$:
\begin{eqnarray*}
\begin{cases}
P_{00} =         \lambda_0 ,\\
P_{0k} =         \lambda_k , \quad k \geq 1 ,
\end{cases}
\end{eqnarray*}
\item For $h \geq 1$:
\begin{align}
\label{transition Pro}
\begin{cases}
\displaystyle P_{h,h-k} =  \sum_{n=0}^{\min(h,L)-k} {\lambda_n \xi_{h,n+k}}   , & 1 \leq k \leq  \min(h,L) , \\
\displaystyle P_{h,k} =  \lambda_0 \xi_{h,0} + \sum_{n=1}^{\min(h,L)} {\lambda_n \xi_{h,n}}  , & k=h   ,  \\
\displaystyle P_{h,h+k} =  \sum_{n=0}^{\min(h,L)} {\lambda_{n+k} \xi_{h,n}}  , & k \geq 1 .
\end{cases}
\end{align}
\end{itemize}

The rationale for the calculation of transition probability is explained as follows:
\begin{itemize}
\item When $h=0$, i.e., there are no backlogged packets in the frame, the backlog size remains $0$ if no new packets arrive and increases by $k$ if $k$ new packets arrive in the frame.
\item When $h>0$, we have three possibilities, corresponding to the cases where the backlog size decreases, remains unchanged and increases, respectively:
    \begin{itemize}
    \item The state $ 1 \leq k \leq \min(h,L)$ corresponds to the case where the backlog size decreases when $n \leq \min(h,L)-k$ new packets arrive and the number of successfully transmitted packets is between $n$ and $\min(h,L)$.
    \item The backlog size remains unchanged if either of two following events happens: (a) no new packets are generated and all the retransmitted backlogged packets fail; (b) $n \leq \min(h,L)$ new packets arrive and the number of the successfully received packets is also $n$.
    \item The backlog size increases if the number of successfully transmitted packets is less than that of new arrivals.
    \end{itemize}
\end{itemize}

In order to establish the ergodicity of the backlog Markov chain $(X_i)_{i \geq 0}$, it is necessary to ensure $(X_i)_{i \geq 0}$ is irreducible and aperiodic. To this end, we conclude this subsection by providing the sufficient conditions on $\{ \lambda_n \}$ for the irreducibility and the aperiodicity of $(X_i)_{i \geq 0}$ as follows:
\begin{eqnarray}
0 < \lambda_n < 1, \ \forall \ n \geq 0.
\label{eq:condition}
\end{eqnarray}

We would like to point out that most of traffic models can satisfy~\eqref{eq:condition}. Throughout the paper, it is assumed that~\eqref{eq:condition} holds and hence $(X_i)_{i \geq 0}$ is irreducible and aperiodic.

\subsection{Stability analysis}

Recalling Definition~\ref{def:sta}, to study the stability of FSA, we need to analyse the ergodicity of the backlog Markov chain $(X_i)_{i \geq 0}$. We first introduce two auxiliary lemmas which will be useful in the ergodicity demonstration.

\begin{lemma}[\cite{Parks1969}]
\label{lemma:pakes}
Given an irreducible and aperiodic Markov chain $(X_i)_{i\geq0}$ such that
\begin{enumerate}
\item the state space is nonnegative integers with the transition probability matrix $ \{P_{hk} \} $,
\item the expected drift at the state $h$, $|D_h|< \infty$, $\forall \ h$,
\item $\lim{\sup_{h\to \infty}}D_h < 0$,
\end{enumerate}
it holds that $(X_i)_{i\geq0}$ is ergodic.
\end{lemma}

\begin{lemma}[\cite{Kaplan1979}]
\label{lemma:kaplan}
Under the assumptions of Lemma \ref{lemma:pakes}, if for some integer $Q\geq0$ and some constant $B \geq 0$, $c \in [0,1]$ such that
\begin{enumerate}
\item $D_h> 0$, for all $h \geq Q$,
\item $\phi^h - \sum_{k}{P_{hk}\phi^h \geq -B(1-\phi)}$, for all $h \geq Q$, $\phi \in [c,1]$,
\end{enumerate}
then $(X_i)_{i\geq0}$ is not ergodic.
\end{lemma}

Armed with Lemma~\ref{lemma:pakes} and Lemma~\ref{lemma:kaplan}, we further prove Theorem~\ref{theorem:stable}.

\begin{proof}[Proof of Theorem \ref{theorem:stable}]
In the proof, we first formulate the expected drift and then study the ergodicity of Markov chain based on drift analysis.

Let $D_h$ be the expected drift of the backlog Markov chain $(X_i)_{i \geq 0}$ at state $h \ (h\geq0)$, formally defined as follows:
\begin{equation}
D_h = E[X_{i+1} - X_i| X_i = h].
\end{equation}
And denote by $C_i$ the number of successful transmissions in frame $i$, we have
\begin{equation*}
X_{i+1} - X_i = N_i - C_i.
\end{equation*}
It then follows that
\begin{equation}
D_h = E[N_i - C_i| X_i = h] = \lambda - E[C_i|X_i].
\label{eq_1}
\end{equation}

Since all of the backlogged packets in frame $i$ are retransmitted, we have
\begin{equation*}
P\{C_i = k | X_i = h\}= \xi_{h, k}, \ 0\leq k \leq \min(h,L),
\end{equation*}
where $\xi_{00} = 0$. Recalling the definition of $r_h$ as the expected number of successfully received packets in the frame, we have
\begin{equation}
E[C_i = k | X_i = h]= r_{h}.
\label{eq_B}
\end{equation}

Following \eqref{eq_1} and \eqref{eq_B}, we obtain the value of the expected drifts as follows:
\begin{equation}
\label{eq:D}
D_h = \lambda -  r_{h}.
\end{equation}

After formulating the expected drift, we then proceed by two steps.

\textbf{Step 1: $L= \Theta(h)$ and $ \Lambda < \alpha e^{-\alpha}$.}

We first show that $|D_h|$ is finite. From \eqref{eq:D}, we have
\begin{eqnarray}
\label{eq:finite}
|D_h| = |\lambda -  r_{h}|  \leq |\lambda| + | r_{h}|   \leq \lambda + \sum_{k=1}^{h} {k \xi_{h,k}}   \leq \lambda + L,
\end{eqnarray}
which demonstrates that $|D_h|$ is finite.

We then derive the limit of $D_h$:
\begin{equation}
\label{eq:lim D}
\lim_{h \to \infty} D_h
= \lambda - \lim_{h \to \infty} r_h  = L \left\{\Lambda - \lim_{h \to \infty} { \binom {h}{1} \frac{1}{L} \left[1 -
    \frac{1}{L}\right]^{h-1}}\right\}
= L (\Lambda - \alpha e^{-\alpha}),
\end{equation}
where $\alpha \triangleq \frac{h}{L}$.
We thus have
\begin{equation*}
\lim_{h \to \infty} D_h  < 0,
\end{equation*}
when $L= \Theta(h)$ and  $ \Lambda < \alpha e^{-\alpha}$.

It then follows from Lemma~\ref{lemma:pakes} that $(X_i)_{i \geq 0}$ is ergodic. Specially, when $\alpha =1$, the system stability region is maximized, i.e., $\Lambda < e^{-1}$.

\textbf{Step 2: $L = o(h)$ or $L = O(h)$ or $L = \Theta(h)$ and $\Lambda > \alpha e^{-\alpha}$.}

In this step, we prove the instability of $(X_i)_{i \geq 0}$ by applying Lemma~\ref{lemma:kaplan}. Taking into consideration the impact of different relation between $L$ and $h$ on the limit of $D_h$, we have:
\begin{itemize}
\item when $L=o(h)$, it holds that
  \begin{equation*}
       \Lambda -\lim_{h \to \infty} { \binom {h}{1} (\frac{1}{o(h)}) (1 - \frac{1}{o(h)})^{h-1}} \\
     = \Lambda - \lim_{h \to \infty} \frac{h}{o(h)} e^{-\frac{h}{o(h)}} = \Lambda >0;
  \end{equation*}
\item when $L=O(h)$, it holds that
  \begin{equation*}
        \Lambda - \lim_{h \to \infty} { \binom {h}{1} (\frac{1}{O(h)}) (1 - \frac{1}{O(h)})^{h-1}} \\
      = \Lambda - \lim_{h \to \infty} \frac{h}{O(h)} e^{-\frac{h}{O(h)}} = \Lambda > 0;
  \end{equation*}
\item when $L=\Theta(h)$ and $\Lambda > \alpha e^{-\alpha}$, it holds that
  $$\Lambda - \lim_{h \to \infty} { \binom {h}{1} \frac{1}{\Theta(h)} \left[1 - \frac{1}{\Theta(h)}\right]^{h-1}}=\Lambda - \alpha e^{-\alpha} >0.$$
\end{itemize}

Substituting these results into~\eqref{eq:lim D}, we have $\displaystyle \lim_{h \to \infty} D_h > 0$, which proves the first condition in Lemma \ref{lemma:kaplan}.

%

Next, we will show that the second condition is also satisfied.
Referring to \cite{Sennott1983}, The key step here is to show that the downward part of the expected drift at state $h \ (h\geq0)$, formally defined as
\begin{equation}
d_{h^-} \triangleq \sum_{k= \max(h-L,0)}^{h} {(k-h) P_{h,k}},
\end{equation}
is lower-bounded. According to the transition probability \eqref{transition Pro}, we have
\begin{equation*}
d_{h^-}       = - \sum_{k=1}^{\min(h,L)} k \sum_{n=0}^{\min(h,L)-k} { \lambda_n \xi_{h,n+k} }.
\end{equation*}

After some algebraic operations, we further get
\begin{align}
\label{eq:down_dr}
d_{h^-}
      & = - \sum_{n=0}^{\min(h,L)} \lambda_n \sum_{k=0}^{\min(h,L)-n} {k \xi_{h,n+k} }
        \geq - \sum_{n=0}^{\min(h,L)} \lambda_n \sum_{k=0}^{\min(h,L)-n}  (k+n) \xi_{h,n+k}
        \geq - \sum_{n=0}^{\min(h,L)} \lambda_n \sum_{k=0}^{\min(h,L)} k \xi_{h,k} \nonumber \\
      & \geq - \sum_{n=0}^{\min(h,L)} \lambda_n r_{h}
        \geq - \min(h,L).
\end{align}

It then holds that $d_{h^-}$ is lower-bounded, which completes the proof of Step 2 and also the proof of Theorem~\ref{theorem:stable}.
\end{proof}

\subsection{System behavior in instability region}

It follows from Theorem~\ref{theorem:stable} that the system is not stable in the following three conditions: $L = o(h)$; $L = O(h)$; and $L = \Theta(h)$ but $\Lambda > \alpha e^{-\alpha}$. In this section, we further investigate the system behavior in the instability region, i.e., when $(X_i)_{i \geq 0}$ is nonergodic. The key results are given in Theorem~\ref{th:NON_oh}.

Before proving Theorem~\ref{th:NON_oh}, we first introduce the following lemma~\cite{Mertens1978} on the conditions for the transience of a Markov chain.

\begin{lemma}[\cite{Mertens1978}]
\label{Le: transi}
Let $(X_i)_{i \geq 0}$ be an irreducible and aperiodic Markov chain with the nonnegative integers as its state space and one-step transition probability matrix $\mathbf{P}=\{P_{hk}\}$. $(X_i)_{i \geq 0}$ is transient if and only if there exists a sequence $\{y_i\}_{i \geq 0}$ such that
    \begin{enumerate}
    \item $y_i$ $(i \geq 0)$ is bounded,
    \item for some $h \geq N$, $y_h < y_0,\ y_1,\ \cdots,\ y_{N-1}$,
    \item for some integer $N>0$, $\sum_{k=0}^{\infty} {y_k P_{hk}} \leq y_h$, $\forall \ h \geq N$.
    \end{enumerate}
\end{lemma}

Armed with Lemma~\ref{Le: transi}, we now prove Theorem \ref{th:NON_oh}.

\begin{proof}[Proof of Theorem~\ref{th:NON_oh}]
The key to prove Theorem~\ref{th:NON_oh} is to show the existence of a sequence satisfying the properties in Lemma~\ref{Le: transi}, so we first construct the following sequence~\eqref{eq:Seq} and then prove that it satisfies the required conditions.
\begin{equation}
\label{eq:Seq}
y_i = \frac{1}{(i+1)^\theta}, \ \theta \in (0,1).
\end{equation}
It can be easily checked that $\{y_i\}$ satisfies the first two properties in Lemma~\ref{Le: transi}.

Next, We distinguish three cases in the rest of proof.

\textbf{Case 1: $L = o(h^{1- \epsilon})$.}

In this case, we first derive an equivalent condition on $\xi_{hk}$ shown in the following lemma for the third property in Lemma~\ref{Le: transi}. The proof is detailed in Appendix~\ref{Appendix:A}.


\begin{lemma}
\label{Le:sup_oh}
If $\lim_{k \to \infty} k^2 \sup_{h \geq k} \xi_{hk} = 0$, then $(X_i)_{i \geq 0}$ is always transient for $L=o(h^{1- \epsilon})$.
\end{lemma}

Lemma \ref{Le:sup_oh} provides another approach to prove the transience of $(X_i)_{i \geq 0}$ in the instability region, i.e., whether $\xi_{hk}$ is of the property that $\lim_{k \to \infty}$ $ k^2 \sup_{h \geq k} \xi_{hk} = 0$ when $L= o(h^{1-\epsilon})$. To this end, we introduce the following lemma by which an approximate result on $\xi_{hk}$ can be derived.

\begin{lemma}[\cite{Feller2008}]
\label{Lem:poisson}
Let $h$ packets each is sent in a slot picked at random among $L$ time-slots in the frame $i$. And Let $h, L \to \infty$ such that $\rho_{_j} = L \frac{e^{-h/L}}{j !} (\frac{h}{L})^j $ remains bounded, then the probability $P(m_j)$ of finding exactly $m_j$ time-slots with $j$ packets can be approximated by the following Poisson distribution with the parameter $\rho_j$,
\begin{equation}
P( m_j) = e^{- \rho_{_j}} \frac{{\rho_{_j}}^{m_j}}{{m_j}!}.
\end{equation}
\end{lemma}

Now, we want to show Lemma \ref{Lem:poisson} is applicable When $L = o(h^{1-\epsilon})$ for $h$ large enough. This is true because
\begin{align}
0 \leq \rho_j \leq \frac{h^j}{j! L^{j-1} e^{h^ \epsilon}} \leq \frac{h^j}{j! L^{j-1}} \cdot \frac{(\lceil \frac{1}{\epsilon} \rceil j)!} {(h^{ \epsilon})^{\lceil \frac{1}{\epsilon} \rceil j} } \leq  \frac{(\lceil \frac{1}{\epsilon} \rceil j)!}{j! {L^{j-1}}},
\label{eq_rho}
\end{align}
meaning that $\rho_j$ is bounded if $j$ is finite.

For FSA-SPR, the probability $\xi_{hk}$ that $k$ packets are received successfully is equivalent to the probability $P(m_j = k)$ when $j=1$, we thus have
\begin{equation}
  \xi_{hk} = e^{- \rho_{_1}} \frac{{\rho_{_1}}^{k}}{k!},
\end{equation}
where $\rho_1 \le \lceil \frac{1}{\epsilon} \rceil !$.

Using the inequality $\frac{n^k}{k!} \leq (\frac{ne}{k})^k$ yields
\begin{equation*}
0 \leq \xi_{hk} \leq e^{- \rho_1} (\frac{\rho_1 e}{k})^k  .
\end{equation*}
Thus,
\begin{equation*}
0 \leq k^2 \xi_{hk} \leq e^{- \rho_1} k^2 (\frac{\rho_1 e}{k})^k  .
\end{equation*}

Furthermore, since $e \rho_1 \leq e {\lceil \frac{1}{\epsilon} \rceil}!$ is constant, we have
\begin{equation*}
\lim_{k \to \infty} {e^{- \rho_1} k^2 (\frac{\rho_1 e}{k})^k} = 0,
\end{equation*}
which means $\xi_{hk}$ in FSA-SPR satisfies the property that $\lim_{k \to \infty}$ $ k^2 \sup_{h \geq k} \xi_{hk} = 0$ when $L = o(h^{1- \epsilon})$. The first case on the transience of $(x_i)_{i \ge 0}$ is proved.


\textbf{Case 2: $L = \Theta (h)$ and $\Lambda > \alpha e^{-\alpha}$.}

In this case, we directly prove that ${y_i}$ satisfies the third property in Lemma~\ref{Le: transi} by algebraic operations. The key steps we need here are to obtain upper bounds of $\xi_{hk}$ and the arrival rate in a new way. To this end, we first recomputed $\xi_{hk}$ as follows when $L = \Theta (h)$:

\begin{align}
\label{eq:Suc_case_2}
\begin{cases}
\displaystyle \xi_{hk} & =    \frac{\binom{L}{k} k! (L-k)^{h-k}}{L^h} - \frac{\binom{L}{k+1} (k+1)! (L-k-1)^{h-k-1}}{L^h} \\
                  & \le  \frac{\binom{L}{k} k!}{L^k} \big( (1-\frac{k}{L})^{h-k} - (1-\frac{k+1}{L})^{h-k} \big) \\
                  &  \le  \frac{\binom{L}{k} k!}{L^k} \le  (1 - \frac{k/2}{L})^{k/2},                                        \;\;\;\;\;\;\;\;\;\;\;\;\;\;\;\;\;\;\;\;\;\;\;\;\;\;\;\;\;\;\;\;\;\;\;\;\;\;\;k < \min(h,L),\\
\displaystyle \xi_{hh} &=    \frac{\binom{L}{h} h! }{L^h} \le  (1 - \frac{h/2}{L})^{h/2}
                   \le (1- \alpha / 2)^ {h/2},                                             \;\;\;\;\;\;\;\;\;\;\;\;\;\;k= h \le L.
\end{cases}
\end{align}

The rational behind the above inequalities is as follows: the event of exactly $k$ successful packets is equivalent to the event that at least $k$ successful packets deducts the part of at least $k+1$ successful packets.


Next, we introduce a lemma to bound the probability distribution of the arrival rate. When the number of new arrivals per slot $A_{il}$ is Poisson distributed with the mean $\Lambda$, the number of new arrivals per frame $N_i$ ($A_{il}$ and $N_i$ is formally defined in Sec.~\ref{sec:SM}.) is also a Poisson random variable with the mean $\lambda = L \Lambda > he^{-\alpha}$.

\begin{lemma}[\cite{Mitzenmacher2005}]
\label{lemma:aux}
Given a Poisson distributed variable $X$ with the mean $\mu$, the following inequality holds:
\begin{equation*}
Pr[X \le x] \le \frac{e^{-\mu}(e \mu)^x}{x^x} ,\ \forall \ x < \mu .
\end{equation*}
\end{lemma}
Applying Lemma \ref{lemma:aux}, we have
\begin{equation}
\label{eq:bound_Poi}
P\{N_i \leq he^{-\alpha}\}
\leq     \frac{e^{-\lambda}(e\lambda)^\frac{h}{e^{\alpha}}}{(\frac{h}{e^{\alpha}})^\frac{h}{e^{\alpha}}}
\leq
         e^{-\frac{h}{e^\alpha} (\frac{L \Lambda}{h e^{-\alpha}} -1)} \left(\frac{L \Lambda}{h e^{-\alpha}}\right)^{\frac{h}{e^{\alpha}} }
         \leq   \frac{1}{a^\frac{h}{e^{\alpha}}},
\end{equation}
where $a \triangleq \frac{h e^{-\alpha}}{L \Lambda} \cdot e^{ \frac{L \Lambda}{h e^{-\alpha}}-1} >1$, following the fact that $e^x > 1+x$, for $\forall \ x>0$.

%
%

Armed with~\eqref{eq:Suc_case_2} and~\eqref{eq:bound_Poi}, we prove that the third property in Lemma~\ref{Le: transi} is also satisfied when $L= \Theta(h)$ by considering different value of $L$.

\textbf{(1) $h \le L$.}

By employing the sequence $\{y_i\}_{i \geq 0}$, we have
\begin{align}
\sum_{k=0}^{\infty} y_k P_{hk}
= & \sum_{k=0}^{\lfloor\frac{h}{e^ \alpha}\rfloor} y_k P_{hk}
                                     + \sum_{k=\lfloor\frac{h}{e^\alpha}\rfloor +1}^{h} y_k P_{hk}
                                     + \sum_{k=h+1}^{\infty} y_k P_{hk} \\
\leq & \sum_{k=0}^{\lfloor\frac{h}{e^ \alpha}\rfloor} \sum_{n=0}^{k}\lambda_n
                                      + y_{h+1}  
                                      + \sum_{k=\lfloor\frac{h}{e^ \alpha}\rfloor +1}^{h} y_k \sum_{n=0}^{k} \lambda_n  (1 - \frac{n+h-k}{2 L})^{(n+h-k)/2} \nonumber\\
\leq &  \frac{\lfloor\frac{h}{e^ \alpha}\rfloor +1}{a^{\lfloor\frac{h}{e^ \alpha}\rfloor}}
                                     + y_{h+1}  
                                     + \sum_{k=\lfloor\frac{h}{e^ \alpha}\rfloor+1}^{h} y_k \sum_{n=0}^{k} \lambda_n  (1 - \frac{n}{2 L})^{\frac{n}{2}} \nonumber\\
\leq &  \frac{\lfloor\frac{h}{e^ \alpha}\rfloor +1}{a^{\lfloor\frac{h}{e^ \alpha}\rfloor}}
                                     + y_{h+1}  
                                     + \sum_{k=\lfloor\frac{h}{e^ \alpha}\rfloor+1}^{h} y_k \big(\sum_{n=0}^{\lfloor\frac{h}{e^ \alpha}\rfloor} \lambda_n  (1 - \frac{n}{2 L})^{\frac{n}{2}}
                                     + \sum_{n=\lfloor\frac{h}{e^ \alpha}\rfloor+1}^{k} \lambda_n  (1 - \frac{n}{2 L})^{\frac{n}{2}} \big) \nonumber\\
\leq & \frac{\lfloor\frac{h}{e^ \alpha}\rfloor +1}{a^{\lfloor\frac{h}{e^ \alpha}\rfloor}}
                                     + \sum_{k=\lfloor\frac{h}{e^ \alpha}\rfloor+1}^{h} y_k  \big( \sum_{n=0}^{\lfloor\frac{h}{e^\alpha}\rfloor} \lambda_n
                                     + (1 - \frac{\frac{h}{e^ \alpha}}{2 L})^{\frac{h}{2 e^ \alpha}} \big)   + y_{h+1} \nonumber\\
\leq  &  \frac{\lfloor\frac{h}{e^ \alpha}\rfloor +1}{a^{\lfloor\frac{h}{e^ \alpha}\rfloor}}
                                     + \frac{h- \lfloor\frac{h}{e^ \alpha}\rfloor}{(\frac{h}{e^ \alpha} +2)^\theta}
                                       \big(\frac{1}{a^\frac{h}{e^ \alpha}}
                                     + (1 - \frac{\alpha }{2 e^ \alpha})^{\frac{h}{2 e^ \alpha}} \big)
                                     +  y_{h+1} 
\leq   \frac{1}{(h+1)^\theta}, \;\;\; \text{as} \; h \to \infty,
\label{eq:Theta_trans}
\end{align}
where the last inequality holds because
\begin{itemize}
\item
$\frac{\lfloor\frac{h}{e^ \alpha}\rfloor +1}{a^{\lfloor\frac{h}{e^ \alpha}\rfloor}}+\frac{h- \lfloor\frac{h}{e^ \alpha}\rfloor}{(\frac{h}{e^ \alpha} +2)^\theta}\big(\frac{1}{a^\frac{h}{e^ \alpha}} + (1 - \frac{\alpha }{2 e^ \alpha})^{\frac{h}{2 e^ \alpha}} \big) \sim \Theta(\frac{h}{{a_0}^{he^{-\alpha}}})$ where $a_0 = \max\left(a, \frac{1}{\sqrt{1-\frac{\alpha}{2 e^\alpha}}}\right)>1$,
\item $y_h - y_{h+1} = \frac{1}{(h+1)^\theta} (1- (1- \frac{1}{h+2})^\theta) \ge \frac{\theta}{(h+1)^\theta(h+2)}$ where we use the fact that $(1-\frac{1}{h+2})^\theta \le 1-\frac{\theta}{h+2}$ following Taylor's theorem.
\end{itemize}

\textbf{(2) $ he^{-\alpha} < L < h$.}

With the same reasoning as~\eqref{eq:Theta_trans} and noticing the fact that at most $L-1$ packets are successfully received when $L < h$, we have
\begin{align}
\sum_{k=0}^{\infty} y_k P_{hk}
=  &    \sum_{k=0}^{h-L-1} y_k P_{hk}  + \sum_{k=h-L}^{h} y_k P_{hk}  + \sum_{k=h+1}^{\infty} y_k P_{hk}
\leq    y_{h+1}  
             + \sum_{k=h-L}^{h} y_k \sum_{n=0}^{L+k-h} \lambda_n  (1 - \frac{n+h-k}{2 L})^{(n+h-k)/2} \nonumber\\
\leq &   y_{h+1}  
                + \sum_{k=h-L}^{h} y_k \sum_{n=0}^{L} \lambda_n  (1 - \frac{n}{2 L})^{\frac{n}{2}}
\leq   y_{h+1}
                + \sum_{k=h-L}^{h} y_k \big( \sum_{n=0}^{\lfloor\frac{h}{e^\alpha}\rfloor} \lambda_n + \sum_{n=\lfloor\frac{h}{e^\alpha}\rfloor +1 }^{L} (1 - \frac{n}{2 L})^{\frac{n}{2}} \big) \nonumber \\
\leq &   \frac{L}{(h-L+1)^\theta} \big(\frac{1}{a^\frac{h}{e^ \alpha}} + (1 - \frac{\alpha }{2 e^ \alpha})^{\frac{h}{2 e^
         \alpha}} \big) + y_{h+1}
\leq   \frac{1}{(h+1)^\theta}, \;\;\; \text{as} \; h \to \infty.
\label{eq:Theta_trans3}
\end{align}

\textbf{(3) $ L \le he^{-\alpha} $.}

With the same reasoning as~\eqref{eq:Theta_trans} and~\eqref{eq:Theta_trans3}, we have
\begin{align}
\sum_{k=0}^{\infty} y_k P_{hk}
= & \sum_{k=0}^{h-L-1} y_k P_{hk} + \sum_{k=h-L}^{h} y_k P_{hk} + \sum_{k=h+1}^{\infty} y_k P_{hk}
\leq    y_{h+1}  
                                      + \sum_{k=h-L}^{h} y_k \sum_{n=0}^{L+k-h} \lambda_n  (1 - \frac{n+h-k}{2 L})^{(n+h-k)/2} \nonumber\\
\leq &   y_{h+1}  
                + \sum_{k=h-L}^{h} y_k \sum_{n=0}^{L} \lambda_n
                \leq   \frac{L}{(h-L+1)^\theta} \cdot \frac{1}{a^{\frac{h}{e^\alpha}}} + y_{h+1}
\leq   \frac{1}{(h+1)^\theta}, \;\;\; \text{as} \; h \to \infty.
\end{align}

Consequently, the third property in Lemma~\ref{Le: transi} holds, which completes the proof of the transience of $(x_i)_{i \ge 0}$ for Case 2.

Next, we proceed with the proof for the third case.

\textbf{Case 3: $L=O(h)$.}

When $L=O(h)$, it is easy to see that the expected number of new arrivals per frame $\lambda = L \Lambda \gg h$. Since $N_i$ is Poisson distributed as mentioned in Case 2 above, by employing~\eqref{eq:bound_Poi}, we also have
\begin{equation}
\label{eq:bound_Poi_2}
P\{N_i \leq h\} \le \frac{1}{{a_1}^{h}},
\end{equation}
where $a_1 \triangleq \frac{h}{L \Lambda} \cdot e^{ \frac{L \Lambda}{ h}-1} \gg 1$, following the fact that $e^x > 1+x$, for $ \forall \ x>0$.

We now prove the transience of $(x_i)_{i \ge 0}$ when $L= O(h)$. To that end, by combining \eqref{eq:Seq} and \eqref{eq:bound_Poi_2}, for $h$ large enough, we have
\begin{align}
\sum_{k=0}^{\infty}  y_k P_{hk}
=  & \sum_{k=0}^{h} y_k P_{hk} +\sum_{k=h+1}^{\infty} y_k P_{hk}
\leq   \sum_{k=0}^{h} \frac{1}{(k+1)^\theta}\sum_{n=0}^{k}\lambda_n
                                     + \frac{1}{(h+2)^\theta} \nonumber\\
\leq  & \sum_{k=0}^{h} \sum_{n=0}^{k}\lambda_n + \frac{1}{(h+2)^\theta}
                  \leq \frac{h+1}{{a_1}^h} + \frac{1}{(h+2)^\theta}
                            \leq   \frac{1}{(h+1)^\theta}, \;\;\; \text{as} \; h \to \infty.
\label{eq:Trans_O}
\end{align}
where the last inequality is true following the same reasoning as~\eqref{eq:Theta_trans}.

Consequently, it follows Lemma~\ref{Le: transi} that the backlog Markov chain $(X_i)_{i \geq 0}$ is also transient when $L=O(h)$, which completes the proof of Theorem~\ref{th:NON_oh}.

\end{proof}

\section{Stability Analysis of Frame Slotted Aloha System with Multipacket Reception}
\label{sec:Sta-MPR}

In this section, we study stability properties of FSA-MPR where the receiver can correctly receive multiple packets transmitted simultaneously. This MPR capability can be achieved through the smart antenna techniques such as MIMO, or coding techniques such as CDMA. Specifically, we exploit MPR-$M$ model that only up to $M$ simultaneous packets can be decoded successfully at the receiver, which differs FSA-MPR from FSA-SPR. The value of $M$ is fixed and is known beforehand.

Note that FSA-SPR investigated previously is a special case of MPR-$M$ with $M=1$, meaning the only difference between FSA-MPR and FSA-SPR is the packet success probability. Hence, with the packet success probability of FSA-MPR, we can develop the analysis on the stability properties of FSA-MPR using similar ideas in the FSA-SPR case. Firstly, conditions for the stability of FSA-MPR are established and we further analyse the system behavior in instability region subsequently.


\subsection{Stability analysis}

We employ Lemma~\ref{lemma:pakes} and Lemma~\ref{lemma:kaplan} as mathematical tools to study the stability properties of FSA-MPR, more specifically, in the proof of Theorem~\ref{theorem:stable-MPR}.
\begin{proof}[Proof of Theorem~\ref{theorem:stable-MPR}]
When $h$ is large, with Poisson approximation, \eqref{eq:Suc_MPR} can be approximated as follows:
\begin{align}
\label{eq:E_sup_MPR}
r_h  = L \sum_{x=1}^{M} e^{-\alpha} \frac{{\alpha ^x}}{(x-1)!}  = L \alpha \sum_{x=0}^{M-1} e^{-\alpha} \frac{{\alpha ^x}}{x!}  \leq L \alpha
\end{align}
where $\alpha \triangleq \frac{h}{L}$. And we define $\Phi (\alpha) \triangleq \sum_{x=1}^{M} e^{-\alpha} \frac{{\alpha ^x}}{(x-1)!}$ in the rest of this paper.

We then develop our proof in three steps.

\textbf{Step 1: $L = \Theta(h)$ and $\Lambda < \sum_{x=1}^{M} e^{-\alpha} \frac{{\alpha ^x}}{(x-1)!}$.}

According to~\eqref{eq:finite} and~\eqref{eq:E_sup_MPR}, the expected drift at state $h$ of $(x_i)_{i\ge 0}$ in FSA-MPR is finite as shown in the following inequality:
\begin{eqnarray*}
|D_h| \leq \lambda + |r_h| \leq \lambda + L \alpha,
\end{eqnarray*}
which demonstrates the second conditions in Lemma~\ref{lemma:pakes} for the ergodicity of $(x_i)_{i\ge 0}$.

Furthermore, when $L= \Theta(h)$, we have
\begin{align*}
\lim_{h \to \infty} D_h  = L (\Lambda - \sum_{x=1}^{M} e^{-\alpha} \frac{{\alpha ^x}}{(x-1)!}) < 0,
\end{align*}
if $\Lambda < \sum_{x=1}^{M} e^{-\alpha} \frac{{\alpha ^x}}{(x-1)!}$.

It then follows from Lemma~\ref{lemma:pakes} that $(x_i)_{i\ge 0}$ is ergodic if $\Lambda < \Phi(\alpha)$ when $L= \Theta(h)$.

\textbf{Step 2: stability region maximiser $\alpha^*$.}

In this step, we show that $\alpha^* =\Theta(M)$. Since the proof consists mainly of algebraic operations of function optimization, we state the following lemma proving Step 2 and detail its proof in Appendix~\ref{Appendix:B}.
\begin{lemma}
\label{Le:alpha_M}
Let $\alpha^*$ denote the value of $\alpha$ that maximises the stability region, it holds that $\alpha^* =\Theta(M)$
\end{lemma}


\textbf{Step 3: $L=o(h)$ or $L=O(h)$ or $L = \Theta(h)$ and $\Lambda > \sum_{x=1}^{M} e^{-\alpha} \frac{{\alpha ^x}}{(x-1)!}$.}

In this step, we prove the instability of $(x_i)_{i\ge 0}$ by applying Lemma~\ref{lemma:kaplan}.

According to the analysis on the impact of different relation between $L$ and $h$ in the the second step proof of Theorem~\ref{theorem:stable}, we know that $\lim_{h \to \infty} D_h > 0$, if the conditions in the second part of Theorem \ref{theorem:stable-MPR} are satisfied.

In addition, similar to \eqref{eq:down_dr}, we can prove that the lower-bounded downward part of expected drift in FSA-MPR as follows:
\begin{equation*}
d_{h^-} \geq -M \min(h,L),
\end{equation*}
which proves the instability of FSA-MPR following Lemma \ref{lemma:kaplan}, which completes the proof of Theorem \ref{theorem:stable-MPR}.
\end{proof}

\subsection{System behavior in instability region}

It follows from Theorem~\ref{theorem:stable-MPR} that the system is not stable in the following three conditions: $L = o(h)$; $L = O(h)$; $L = \Theta(h)$ but $\Lambda > \sum_{x=1}^{M} e^{-\alpha} \frac{{\alpha ^x}}{(x-1)!}$. In this subsection, we further investigate the system behavior in the instability region, i.e., when $(X_i)_{i \geq 0}$ is nonergodic. The key results are given in Theorem~\ref{th:NON_oh_OH-MPR}, whose proof is detailed as follows.

\begin{proof}[Proof of Theorem~\ref{th:NON_oh_OH-MPR}]
We prove the theorem in three cases as the proof of Theorem~\ref{th:NON_oh}.

%

\textbf{Case 1: $L=o(h^{1- \epsilon})$.}

Recall Lemma~\ref{Le:sup_oh}, we can demonstrate the transience of $(x_i)_{i \ge 0}$ by proving $\lim_{k \to \infty}k^2 \xi_{hk} = 0$ in FSA-MPR. To this end, we derive $\xi_{hk}$.

In~\eqref{eq_rho}, we have shown that $\rho_{_j}$ is bounded when $L = o(h^{1- \epsilon})$ for $j=1, 2, \cdots, M$, so the probability of finding exactly $m_j$ time-slots with $j$ packets in FSA-MPR can also be approximated by the Poisson distribution with the parameter $\rho_j$, following from Lemma~\ref{Lem:poisson} .

Consequently, we can derive the probability that all $h$ packets fail to be received, i.e., there is not any slot with $1 \le j \le M$ packets, as follows:
\begin{equation}
\xi_{h0}  =   e^{- (\rho_1+\rho_2+\cdots+\rho_{_M}) }   .
\end{equation}

Further, we can get the following inequalities:
\begin{align*}
\xi_{hk}  \leq 1-e^{- {(\rho_1+\rho_2+\cdots+\rho_{_M})}}  \leq 1-e^{- M \rho_{_M}}, \  1 \le k < L,
\end{align*}
where we use the fact that the probability of exact $k \geq 1$ successfully received packets among $h$ packets is less than that of at least one packet received successfully in the first inequality. And the second inequality above follows from the fact that when $L=o(h^{1- \epsilon})$, it holds that
$$ \rho_{_M} > \rho_{_{M-1}} > \cdots > \rho_2 > \rho_1.$$

As a result, we have
\begin{align*}
0 \leq \lim_{k \to \infty}  k^2 \sup_{h \geq k} \xi_{hk}
&    =  \lim_{k \to \infty} k^2 (1-e^{- M \rho_M})
 \leq  \lim_{h \to \infty} h^2 (1-e^{- M \rho_M})
 \leq  \lim_{h \to \infty} \frac{e^{ M \rho_{M} } - 1}{(1/h^2)} \\
& \leq  \lim_{h \to \infty} \frac{(M+3)(M+2)(M+1)M L^3}{2 e^{h/L}}
  \leq 0,
\end{align*}
where we use L'Hospital's rule and the fact that $\rho_{_M} \to 0$ at an exponential rate as $h \to \infty$.

Thus, according to Lemma \ref{Le:sup_oh}, the backlog Markov chain $(x_i)_{i \ge 0}$ is transient when $L=o(h^{1- \epsilon})$.

\textbf{Case 2: $L=\Theta(h)$ and $ \Lambda > \alpha$.}

In this case, we have $\lambda > L \Lambda >h$.
Using similar reasoning as~\eqref{eq:Trans_O}, we have :
\begin{align*}
\sum_{k=0}^{\infty}  y_k P_{hk}
                            \leq   \frac{h+1}{{a_2}^h} + \frac{1}{(h+2)^\theta}
                            \leq   \frac{1}{(h+1)^\theta}, \ \text{as} \ h \to \infty ,
\end{align*}
where $a_2 = \frac{\alpha}{\Lambda} e^{\frac{\Lambda}{\alpha}-1} >1$.

Thus when $L=\Theta(h)$ and $ \Lambda > \alpha$, $(x_i)_{i \ge 0}$ is transient.

\textbf{Case 3: $L=O(h)$.}

In the analysis of FSA-SPR system, we have proven that when $L = O (h)$, the Markov chain $(x_i)_{i \ge 0}$ is always transient, independent of $\xi_{hk}$. Noticing that $\xi_{hk}$ is the only difference between FSA-SPR and FSA-MPR, it holds that $(x_i)_{i \ge 0}$ is also transient under FSA-MPR.
\end{proof}

\section{Conclusion}
\label{sec:conclusion}
In this paper, we have studied the stability of FSA-SPR and FSA-MPR by modeling the system backlog as a Markov chain. By employing drift analysis, we have obtained the closed-form conditions for the stability of FSA and shown that the stability region is maximised when the frame length equals the backlog size in the single packet reception model and when the ratio of the backlog size to frame length equals in order of magnitude the maximum multipacket reception capacity in the multipacket reception model. Furthermore, to characterise system behavior in the instable region, we have mathematically demonstrated the existence of transience of the Markov chain. Our results provide theoretical guidelines on the design of stable FSA-based protocols in practical applications such as RFID and M2M systems.

\bibliographystyle{abbrv}
\bibliography{stability_frame_SPR_MPR_sc}

\begin{thebibliography}{10}

\bibitem{Abramson1970}
N.~Abramson.
\newblock The aloha system: Another alternative for computer communications.
\newblock In {\em Proceedings of Fall 1970 AFIPS fall joint computer
  conference}, pages 281--285, 1970.

\bibitem{Barletta2012}
L.~Barletta, F.~Borgonovo, and M.~Cesana.
\newblock A formal proof of the optimal frame setting for dynamic-frame aloha
  with known population size.
\newblock arXiv preprint arxiv:1202.3914.

\bibitem{Borst2008}
S.~Borst, M.~Jonckheere, and L.~Leskela.
\newblock Stability of parallel queueing systems with coupled service rates.
\newblock {\em Discrete Event Dynamic Systems}, 18(4):447--472, 2008.

\bibitem{EPCGlobal2005}
{EPCglobal Inc.}
\newblock Radio-frequency identity protocols class-1 generation-2 {UHF} {RFID}
  protocol for communications at 860 mhz - 960 mhz version 1.0.9.
\newblock {\em EPCglobal Inc.}, 17, 2005.

\bibitem{Feller2008}
W.~Feller.
\newblock {\em An introduction to probability theory and its applications},
  volume~2.
\newblock John Wiley \& Sons, 2008.

\bibitem{Ghez1988}
S.~Ghez, S.~Verdu, and S.~C. Schwartz.
\newblock Stability properties of slotted aloha with multipacket reception
  capability.
\newblock {\em IEEE Transactions on Automatic Control}, 33(7):640--649, 1988.

\bibitem{Ghez1989}
S.~Ghez, S.~Verdu, and S.~C. Schwartz.
\newblock Optimal decentralized control in the random access multipacket
  channel.
\newblock {\em IEEE Transactions on Automatic Control}, 34(11):1153--1163,
  1989.

\bibitem{johnson1977}
N.~Johnson and S.~Kotz.
\newblock {\em Urn models and their application: an approach to modern discrete
  probability theory}.
\newblock Wiley, 1977.

\bibitem{Kaplan1979}
M.~Kaplan.
\newblock A sufficient condition for nonergodicity of a markov chain.
\newblock {\em IEEE Transection on Information Theory}, 25(4):470--471, 1979.

\bibitem{Mertens1978}
J.~F. Mertens, E.~S. Cahn, and S.~Zamir.
\newblock Necessary and sufficient conditions for recurrence and transience of
  markov chains, in terms of inequalities.
\newblock {\em Journal of Applied Probability}, 15(4):848--851, 1978.

\bibitem{Mitzenmacher2005}
M.~Mitzenmacher and E.~Upfal.
\newblock {\em Probability and computing: Randomized algorithms and
  probabilistic analysis}.
\newblock Cambridge University Press, 2005.

\bibitem{Okada1977}
H.~Okada, Y.~Igarashi, and Y.~Nakanishi.
\newblock Analysis and application of framed aloha channel in satellite packet
  switching networks-fadra method.
\newblock {\em Electronics and Communications in Japan}, 60:72--80, 1977.

\bibitem{Parks1969}
A.~G. Pakes.
\newblock Some conditions for ergodicity and recurrence of markov chains.
\newblock {\em Operations Research}, 17(6):1058--1061, 1969.

\bibitem{Pountourakis1992}
I.~E. Pountourakis and E.~D. Sykas.
\newblock Analysis, stability and optimization of aloha-type protocols for
  multichannel networks.
\newblock {\em Computer Communications}, 15(10):619--629, 1992.

\bibitem{Rao1988}
R.~R. Rao and A.~Ephremides.
\newblock On the stability of interacting queues in a multiple-access system.
\newblock {\em IEEE Transactions on Information Theory}, 34(5):918--930, 1988.

\bibitem{Lawrence1975}
L.~G. Robert.
\newblock Aloha packet system with and without slots and capture.
\newblock {\em ACM SIGCOMM Computer Communication Review}, 5(2):28--42, 1975.

\bibitem{Rosenkrantz1983}
W.~A. Rosenkrantz and D.~Towsley.
\newblock On the instability of the slotted aloha multiaccess algorithm.
\newblock {\em IEEE transactions on automatic control}, 28(10):994--996, 1983.

\bibitem{Sant2000}
J.~Sant and V.~Sharma.
\newblock Performance analysis of a slotted-aloha protocol on a capture channel
  with fading.
\newblock {\em Queueing Systems}, 34(1-4):1--35, 2000.

\bibitem{Schoute1983}
F.~C. Schoute.
\newblock Dynamic frame length aloha.
\newblock {\em IEEE Transactions on Communications}, 31(4):565--568, 1983.

\bibitem{Sennott1983}
L.~I. Sennott, P.~A. Humblet, and R.~L. Tweedie.
\newblock Mean drifts and the non-ergodicity of markov chains.
\newblock {\em Operations Research}, 31(4):783--789, 1983.

\bibitem{Szpankowski1983}
W.~Szpankowski.
\newblock Packet switching in multiple radio channels: Analysis and stability
  of a random access system.
\newblock {\em Computer Networks}, 7(1):17--26, 1983.

\bibitem{Szpankowski1994}
W.~Szpankowski.
\newblock Stability conditions for some distributed systems: Buffered random
  access systems.
\newblock {\em Advances in Applied Probability}, 26(2):498--515, 1994.

\bibitem{Tsybakov1979}
B.~S. Tsybakov and V.~A. Mikhailov.
\newblock Ergodicity of a slotted aloha system.
\newblock {\em Problemy Peredachi Informatsii}, 15(4):73--87, 1979.

\bibitem{Vasudevan2009}
S.~Vasudevan, D.~Towsley, D.~Goeckel, and R.~Khalili.
\newblock Neighbor discovery in wireless networks and the coupon collector's
  problem.
\newblock In {\em Proceedings of the 15th annual international conference on
  Mobile computing and networking}, pages 181--192, 2009.

\bibitem{Vogt2002}
H.~Vogt.
\newblock Efficient object identification with passive rfid tags.
\newblock In {\em International Conference on Pervasive Computing}, pages
  98--113, 2002.

\bibitem{Wieselthier1989}
J.~E. Wieselthier, A.~Ephremides, and A.~Larry.
\newblock An exact analysis and performance evaluation of framed aloha with
  capture.
\newblock {\em IEEE Transactions on Communications}, 37(2):125--137, 1989.

\bibitem{Wu2013}
H.~Wu, C.~Zhu, R.~J. La, and X.~Liu.
\newblock Fasa: Accelerated s-aloha using access history for event-driven m2m
  communications.
\newblock {\em IEEE/ACM Transactions on Networking}, 21(6):1904--1917, 2013.

\bibitem{Zhu2010}
L.~Zhu and T.~P. Yum.
\newblock Optimal frame aloha{-}based anti{-}collision algorithm for rfid
  systems.
\newblock {\em IEEE Transactions on communications}, 58(12):3583--3592, 2010.

\bibitem{Zhu2011}
L.~Zhu and T.~P. Yum.
\newblock A critical survey and analysis of rfid anti{-}collision mechanisms.
\newblock {\em IEEE Communications Magazine}, 49(5):214--221, 2011.

\end{thebibliography}

\appendices
\section{Proof of Lemma~\ref{Le:sup_oh}}
\label{Appendix:A}
\begin{proof}
Recall the definition of transition probability and we can get the following equivalent:
\begin{align*}
\sum_{k} {y_k P_{hk}} \leq y_h
\Longleftrightarrow  \sum_{k=1}^{\min(h,L)} (y_{h-k} - y_h) P_{h,h-k}
                    + \sum_{k=1}^{\infty} (y_{h+k} - y_h) P_{h,h+k}  \leq 0.
\end{align*}


Define $f'(h)$ and $g'(h)$ as follows:
\begin{align}
\label{eq:g'f'}
\begin{cases}
f'(h) = (h+1)^\theta & \sum_{k=1}^{\min(h,L)}  (\frac{1}{h+1-k} - \frac{1}{h+1})  \\
        &\cdot \sum_{n=0}^{\min(h,L)-k} {\lambda_n \xi_{h,n+k}} , \\
g'(h) = (h+1)^\theta & \sum_{k=1}^{\infty}  (\frac{1}{h+1+k} - \frac{1}{h+1}) \\
        &\cdot \sum_{n=0}^{\min(h,L)} {\lambda_{n+k} \xi_{h,n}}.
\end{cases}
\end{align}

We have
\begin{equation*}
\sum_{k} {y_k P_{hk}} \leq y_h \Longleftrightarrow   f'(h) + g'(h) \leq 0.
\end{equation*}

Furthermore, the expected drift of $X_i$ at state $h$ can be also computed from the one-step transition probabilities~\eqref{transition Pro},
\begin{equation*}
D_h = - \sum_{k=1}^{\min(h,L)} k P_{h,h-k} + \sum_{k=1}^{\infty} k P_{h,h+k} = f(h) + g(h),
\end{equation*}
where $f(h)$ and $g(h)$ are defined as follows:
\begin{align}
\label{eq:gf}
\begin{cases}
f(h) = \sum_{k=1}^{\min(h,L)} k \sum_{n=0}^{\min(h,L)-k} {\lambda_n \xi_{h,n+k}} ,\\
g(h) = \sum_{k=1}^{\infty}  k \sum_{n=0}^{\min(h,L)} {\lambda_{n+k} \xi_{h,n}}  .
\end{cases}
\end{align}

It is noted that in the proof of Theorem~\ref{theorem:stable}, we have show that $\lim_{h \to \infty} D_h$ $ \geq 0 $ in the nonergodicity region. Consequently, it then holds that $(X_i)_{i \geq 0}$ is transient if the following equation holds:
\begin{equation}
\lim_{h \to \infty} [f'(h) + g'(h) + \theta D_h] = 0.
\label{eq:aux3}
\end{equation}

Noticing that $D_h=f(h)+g(h)$, we prove \eqref{eq:aux3} by showing that: (1) $\lim_{h \to \infty} [f'(h)+ \theta f(h)]=0$; (2) $ \lim_{h \to \infty} [g'(h)+ \theta g(h)]=0$.

\textbf{We first prove $\lim_{h \to \infty} [f'(h)+ \theta f(h)]=0$.}

From \eqref{eq:g'f'} and \eqref{eq:gf}, we get
\begin{align}
f'(h)+ \theta f(h) =  (h+1)  \sum_{k=1}^{\min(h,L)} \Big[(\frac{h+1}{h+1-k})^\theta -1
                       -\frac{\theta k}{h+1}\Big] \cdot \sum_{n=0}^{\min(h,L)-k} {\lambda_n \xi_{h,n+k}},
\end{align}
which is nonnegative since
$$ (\frac{h+1}{h+1-k})^\theta -1  \nonumber  -\frac{\theta k}{h+1} > 0, \ \forall 1 \leq k \leq h.$$
Define $\tau_k = \sup_{h \geq k} \xi_{hk}$, we have
\begin{align*}
0
\leq
       f'(h)+ \theta f(h)
&\leq   (h+1)    \sum_{k=1}^{\min(h,L)}  \Big[(\frac{h+1}{h+1-k})^\theta -1 -\frac{\theta k}{h+1}\Big]
                   \sum_{n=0}^{\min(h,L)-k} {\lambda_n \tau_{n+k}} \\ 
& \leq  (h+1)    \sum_{n=0}^{\infty} \lambda_n \sum_{k=1}^{h} \Big[(\frac{h+1}{h+1-k})^\theta -1  -\frac{\theta k}{h+1}\Big]  \tau_{n+k},
\end{align*}
which can also be written as
\begin{equation}
f'(h)+ \theta f(h) \leq
                        z_{1}(h) + z_{2}(h)
\end{equation}
where $z_{1}(h)$ and $z_{2}(h)$ are defined as follows:
\begin{align*}
\begin{cases}
\displaystyle z_{1}(h)=
           (h+1) \sum_{n=0}^{\infty} \sum_{k=1}^{\lfloor \frac{h+1}{2} \rfloor} \Big[(\frac{h+1}{h+1-k})^\theta -1  -\frac{\theta k}{h+1}\Big] \cdot \lambda_n{ \tau_{n+k}},   \\
\displaystyle  z_{2}(h)=
           (h+1) \sum_{n=0}^{\infty} \sum_{k= \lfloor \frac{h+3}{2} \rfloor}^{h} \Big[(\frac{h+1}{h+1-k})^\theta -1 -\frac{\theta k}{h+1}\Big] \cdot \lambda_n { \tau_{n+k}}.
\end{cases}
\end{align*}

Next we show that $z_{1}(h)$ and $z_{2}(h)$ go to zero independently. Given $0< x \leq h$, define $m_h(x)$ as follows:
\begin{equation*}
m_{h} (x) = \frac{h+1}{x^2} \Big[(\frac{h+1}{h+1-x})^\theta -1 \Big] -\frac{\theta }{x},
\end{equation*}
which is positive and nondecreasing in $x$ for $h \geq 1$. We also have
\begin{equation*}
m_{h} \left(\lfloor \frac{h+1}{2} \rfloor \right) \leq  \frac{1}{h+1} [4(2^\theta - 1) - 2 \theta] = \frac{A}{h+1},
\end{equation*}
where $A$ is a positive constant only depending on $\theta$. Thus,
\begin{align*}
z_{1} (h)
=    \sum_{n=0}^{\infty} \lambda_n \sum_{k=1}^{\lfloor \frac{h+1}{2} \rfloor} k^2  m_h (k)  { \tau_{n+k}}
\leq  \frac{A}{h+1} \sum_{n=0}^{\infty} \lambda_n  \sum_{k=1}^{\lfloor \frac{h+1}{2} \rfloor} k^2   {
        \tau_{n+k}}
\leq  \frac{A}{h+1} \sum_{n=0}^{\infty} \lambda_n  \sum_{k=1}^{n+ \lfloor \frac{h+1}{2} \rfloor} k^2  { \tau_{n+k}}.
\end{align*}

If the condition $\lim_{k \to \infty} k^2 \sup_{h \geq k} \xi_{hk} = 0$ holds, i.e., $lim_{k \to \infty} k^2 \tau_{k} = 0$, then it holds that $lim_{n \to \infty} \frac{1}{n} \sum_{k=1}^{n} k^2 \tau_{k} = 0$. Consequently, for $\forall \ \epsilon > 0$, we can choose $h$ large enough such that $\sum_{k=1}^{n} k^2 \tau_{k} < n \epsilon$ for $n \geq \lfloor \frac{h+1}{2} \rfloor$. It holds that
\begin{align*}
z_{1} (h) \leq  \epsilon \frac{A}{h+1} \sum_{n=0}^{\infty} \lambda_n  (n + \frac{h+1}{2})  = \epsilon A \big(\frac{\lambda}{h+1} + \frac{1}{2} \big)\rightarrow 0.
\end{align*}

On the other hand, if $h$ is large enough such that for $\forall \ k \geq \lfloor (h+3)/{2} \rfloor$, it holds that $\tau_k < \epsilon / {k^2}$, we have
\begin{align*}
z_2 (h)
\leq &
      \epsilon \sum_{n=0}^{\infty} \lambda_n \sum_{k=\lfloor \frac{h+3}{2} \rfloor}^{h}
       \Big[ (\frac{h+1}{h+1-k})^\theta -1  -\frac{\theta k}{h+1}\Big]  \frac{h+1}{(n+k)^2}
\leq
      \epsilon \sum_{n=0}^{\infty} \lambda_n \sum_{k=\lfloor \frac{h+3}{2} \rfloor}^{h} \Big[ (\frac{h+1}{h+1-k})^\theta -1  -\frac{\theta k}{h+1}\Big] \frac{h+1}{(\lfloor \frac{h+3}{2} \rfloor)^2} \\
\leq &
       \frac{4 \epsilon}{h+1} \sum_{k=\lfloor \frac{h+3}{2} \rfloor}^{h} \Big[ (\frac{h+1}{h+1-k})^\theta -1 - \frac{\theta k}{h+1}\Big]
\leq
       \frac{2 \theta (1+\theta) \epsilon}{h+1} \sum_{k=\lfloor \frac{h+3}{2} \rfloor}^{h} (\frac{k}{h+1})^2 \leq  \theta (1+\theta) \epsilon,
\end{align*}
where we use the following inequalities
\begin{equation}
0
\leq
     \frac{1}{(1+x)^ \theta} -1 + \theta x \leq \theta (1+ \theta) \frac{x^2}{2}, \;\;\;\;(x \geq 0, 0 < \theta <1).
\label{eq:sup_k_theta}
\end{equation}

Therefore, it holds that $ \lim_{h \to \infty} [f'(h)+ \theta f(h)]=0$.

\textbf{We then prove $ \lim_{h \to \infty} [g'(h)+ \theta g(h)]=0$.}

From \eqref{eq:g'f'} and \eqref{eq:gf}, we get
\begin{equation*}
g'(h)+ \theta g(h)
=  (h+1)  \sum_{k=1}^{\infty} \Big[(\frac{h+1}{h+1+k})^\theta -1   +\frac{\theta k}{h+1}\Big]
   \cdot \sum_{n=0}^{\min(h,L)} {\lambda_{n+k} \xi_{h,n}}.
\end{equation*}
Since $\Big[(\frac{h+1}{h+1+k})^\theta -1  \nonumber +\frac{\theta k}{h+1}\Big] \geq 0$, after some algebraic operations, we have
\begin{equation*}
g'(h)+ \theta g(h)
\leq   (h+1)\sum_{n=1}^{\infty} \lambda_{n} \sum_{k=1}^{n} \Big[(\frac{h+1}{h+1+k})^\theta -1
      +\frac{\theta k}{h+1}\Big] \cdot  {\xi_{h,n-k}}
\end{equation*}
Using the inequalities~\eqref{eq:sup_k_theta}, we have
\begin{align*}
0
\leq
       g'(h)+ \theta g(h)
\leq &   \theta \frac{(\theta+1)}{2} (h+1)\sum_{n=1}^{N} \lambda_{n} \sum_{k=1}^{n}\frac{k^2}{(h+1)^2} {\xi_{h,n-k}}
          + \theta (h+1) \sum_{n=N+1}^{\infty} \lambda_{n} \sum_{k=1}^{n} \frac{k}{h+1} {\xi_{h,n-k}}  \\
\leq &  \frac{1}{h+1} \sum_{n=1}^{N} n^2 \lambda_{n}  + \sum_{n=N+1}^{\infty} n \lambda_{n}.
\end{align*}
For $\forall \ \epsilon >0 $, we choose $N$ such that $\sum_{n=N+1}^{\infty} n \lambda_{n} < \epsilon/2$; then choose $h$ large enough so that $\frac{1}{h+1} \sum_{n=1}^{N} n^2 \lambda_{n} < \epsilon/2$; we have $\lim_{h \to \infty} [g'(h)+ \theta g(h)]=0$, which completes the proof of the second part and also Lemma~\ref{Le:sup_oh}.
\end{proof}

\section{Proof of Lemma~\ref{Le:alpha_M}}
\label{Appendix:B}
\begin{proof}
We write $ \Phi (\alpha)$ as
\begin{align*}
\Phi(\alpha) = e^{-\alpha} \sum_{i=1}^M \frac{\alpha^i}{(i-1)!},
\end{align*}
whose derivative can be calculated as:
\begin{align*}
\Phi'(\alpha) = e^{-\alpha}  \left[\sum_{i=0}^{M-1}\frac{\alpha^i}{i!} - \frac{\alpha ^{M}}{(M-1)!} \right].
\end{align*}

We distinguish two cases.

\begin{itemize}
\item \textit{Case 1: $ \alpha \geq M $}. Note that $ N! \leq N^{N-1}$ $\forall N\in \mathbb{N}$, we have
\begin{equation*}
\Phi'(\alpha)
<   \frac{e^{-\alpha}}{(M-1)!} \cdot  \left( \sum_{i=1}^M M^{M-i}\alpha^{i-1}-\alpha^M\right)
<   \frac{e^{-\alpha}}{(M-1)!} \cdot \big( M \alpha ^{M-1} - {\alpha ^{M}}  \big)
<    0,
\end{equation*}
meaning that $\Phi(\alpha)$ is monotonously decreasing when $ \alpha \geq M $.
\item \textit{Case 2: $ \alpha \leq \frac{M-1}{e}$}.
Using the inequality $N! \geq ({\frac{N}{e}})^{N}$, we have
\begin{multline*}
\Phi'(\alpha)
\geq     \frac{e^{-\alpha}}{(M-1)!} \cdot  \left[ (\frac{M-1}{e})^{M-1} + \alpha(\frac{M-1}{e})^{M-1}
           -{\alpha ^{M}}\right]  \geq    \frac{e^{-\alpha}}{(M-1)!} \cdot \big( \alpha ^{M-1} + \alpha ^{M}- {\alpha ^{M}}  \big)
>  0 ,
\end{multline*}
meaning that $\Phi(\alpha)$ is monotonously increasing when $ \alpha \leq \frac{M-1}{e}$.
\end{itemize}

Combining the analysis in both cases, we have proved that $\alpha^*$ maximising $\Phi(\alpha)$ falls into the interval $[\frac{M-1}{e}, M]$, i.e., $\alpha^* = \Theta (M)$.
\end{proof}

\end{document}